\documentclass[11pt]{article}

\usepackage{amsmath,amsfonts,amssymb,amsthm}
\usepackage{graphicx}

\usepackage{xspace}
\usepackage{tikz}
\usepackage{epstopdf}
\usepackage{subcaption}
\newfont{\ninept}{ptmr at 9pt} 

\usepackage{fullpage,hyperref}
\usepackage{mathtools}

\DeclareMathOperator{\tr}{tr}

\newtheorem{Thm}{Theorem}

\newtheorem{Propn}[Thm]{Proposition}

\newcommand{\Amap}{\bs {\sf A}}

\newcommand{\scp}[2]{\langle #1,\, #2 \rangle}
\newcommand{\inv}[1]{\frac{1}{#1}}
\newcommand{\sinv}[1]{\text{\raisebox{.5mm}{\scalebox{.7}{$\frac{1}{#1}$}}}}
\newcommand{\supp}{{\rm supp}\,}
\newcommand{\tinv}[1]{{\textstyle\frac{1}{#1}}}
\newcommand{\ud}{\mathrm{d}} 

\newcommand{\bs}{\boldsymbol}
\newcommand{\bb}{\mathbb}
\newcommand{\cl}{\mathcal}

\newcommand{\whp}{\mbox{w.h.p.\@}\xspace}

\newcommand{\ts}{\textstyle}

\newcommand{\ie}{\emph{i.e.}, }
\newcommand{\eg}{\emph{e.g.}, }
\newcommand{\sq}{\vspace{-1.5mm}}
\newcommand{\hsq}{\vspace{-0.5mm}}
\newcommand{\im}{{\sf i}\mkern1mu} 

\newcommand{\iid}{%
    \ifmmode
        \mathrm{i.i.d.}%
    \else%
        i.i.d.\@\xspace%
    \fi%
}

\newcommand{\myp}[1]{\noindent\textbf{#1} }

\DeclareMathOperator{\diag}{diag}

\title{1-bit Localization Scheme for Radar using \\Dithered Quantized Compressed Sensing\vspace{0mm}}
\author{Thomas Feuillen\footnote{ TF and LV are with CoSyGroup, ICTEAM/ELEN, Universit\'e catholique de Louvain.}, Chunlei Xu$^\dag$, Luc Vandendorpe$^*$, and
    Laurent Jacques\footnote{CX and LJ are with ISP Group, ICTEAM/ELEN, Universit\'e catholique de Louvain, funded by the F.R.S.-FNRS. Part of this study is funded by the project {\sc AlterSense} (MIS-FNRS). E-mail: \url{{thomas.feuillen,chunlei.xu,luc.vandendorpe,laurent.jacques}@uclouvain.be}.}}

\begin{document}
\maketitle

\begin{abstract}
We present a novel scheme allowing for 2D target localization using highly quantized 1-bit measurements from a \textit{Frequency Modulated Continuous Wave} (FMCW) radar with two receiving antennas. Quantization of radar signals introduces localization artifacts, we remove this limitation by inserting a \emph{dithering} on the unquantized observations. We then adapt the projected back projection algorithm to estimate both the range and angle of targets from the dithered quantized radar observations, with provably decaying reconstruction error when the number of observations increases. Simulations are performed to highlight the accuracy of the dithered scheme in noiseless conditions when compared to the non-dithered and full 32-bit resolution under severe bit-rate reduction. Finally, measurements are performed using a radar sensor to demonstrate the effectiveness and performances of the proposed quantized dithered scheme in real conditions.
\end{abstract}


\section{Introduction}
Compressive sensing aims at compressively and non-adaptively sampling structured signals, \eg sparse or compressible signals in an appropriate basis, by correlating them with a few random patterns, \ie much less numerous than the ambient signal dimension~\cite{candes2006stable}. The compressively observed signal is then estimated from non-linear algorithms such as basis pursuit denoise (BPDN) \cite{CT2005}, iterative hard thresholding (IHT) \cite{BD2009}, or CoSaMP \cite{NT2009}. 

In radar processing, CS offers the potential to simplify the acquisition process \cite{bara} or to use super resolution algorithms to solve ambiguous estimation problems \cite{strohmer}. However, the underlying assumption of such schemes is the availability of high resolution radar signals, requiring high bit-rate data transmission to a processing unit. 

In this article, we aim to break this assumption and to further explore the reconstruction of the target scene on the basis of radar signals acquired under harsh bit-rate acquisition process, \ie a regime where classic estimation methods fail (\eg Maximum Likelihood~\cite{mkay}). Bit-rate reduction in radar applications indeed opens new study directions, \eg through the use of 1-bit comparators to design cost-efficient acquisition hardware, or the use of several radar sensors run in parallel with fixed data-rate, as in \textit{Internet of Things} (IOT) applications relying on massive collection of sensors. Moreover, this loss of resolution can be counteracted by increasing the number of observations, provided that new algorithms be designed for this context. 
 
We propose to reconstruct the target scene in the extreme 1-bit measurement regime, in a similar way to only recording the sign of each sample~\cite{BB2008, JLBB2013, PV2013}. To the best of our knowledge, our paper is one of a few initial works addressing the case of 1-bit FMCW radar processing. Comparing with the existing literature on 1-bit quantization of ``IQ'' signals for different radar applications \cite{SARMAP,varthres}, our main contributions lie in the following aspects. First, we show that estimating the 2D-localization of multiple targets observed from a radar system with two antennas under the harsh bit-rate requirement is feasible. This problem amounts to estimating a sparse signal, whose support and phases encode the target ranges and angles, from a quantized CS (QCS) model. In particular, we explore the estimation in an extreme bit-rate scenario where every measurement takes a single bit achieved by a uniform scalar quantization combined with a random dithering vector~\cite{DJR2017,XJ2018, XSJ2018}. 

Second, we provide theoretical guarantees on the estimation error of multiple targets localization using the \textit{projected back projection} (PBP) algorithm \cite{BFNPW2017, XJ2018, XSJ2018}. This is achieved by promoting in PBP a joint support between the range profiles observed by the two antennas. In particular, we show that the estimation error decays when the number of quantized observations increases. We further reveal, through Monte Carlo simulations, a certain trade-off between the number of measurements and the total bit-rate by comparing the performances of PBP under multiple scenarios involving one or two targets and different measurement numbers and resolutions. The importance of the dithering process is also stressed by the existence of strong artifacts in the 2D-localization of targets when this dithering is not added. Finally, we demonstrate our method in real experiments, locating corner reflectors in an anechoic chamber. In this context, we show that the random dithering still improves the localization of targets provided this dithering is adapted to the signal noise.

The rest of this paper is structured as follows. The radar signal model and the simplified linear model are introduced in Sec.~\ref{sec:Radar-system}. The quantized radar observation model, the adaptation of the PBP algorithm to the 2D-localization of multiple targets and the theoretical analysis of its reconstruction error are provided in Sec.~\ref{sec: QCS-model}. In Sec.~\ref{sec:num-res}, the proposed scheme is tested under different scenarios using Monte Carlo simulations. Finally, we report the use of our framework in a real experiment in Sec. \ref{sec:real-experiment} before to conclude. 

\myp{Notations:} Vectors and matrices are denoted with bold symbols. The imaginary unit is $\im = \sqrt{-1}$, $\angle (r e^{\im \alpha}) = \alpha$, $\bb Z_\delta:=\delta \bb Z + \delta/2$, $[D] := \{1, \cdots, D\}$ for $D \in \bb N$, ${\rm \bf Id}$ is the identity matrix, $\supp \bs u = \{i : u_i \neq 0\}$ is the support of $\bs u$, $\lfloor \cdot \rfloor$ is the flooring operator, and $|\cl S|$ is the cardinality of a set $\cl S$. For any complex quantity $B$, \eg a scalar, a vector or a matrix, $B_{\rm R} = \Re(B)$ and $B_{\rm I} = \Im(B)$ are the real and imaginary parts of $B$, respectively. For any $\bs B \in \bb C^{D \times D'}$ (or $\bs u \in \bb C^{D'}$), $\bs B_{\cl S}$ (resp. $\bs u_{\cl S}$) is the cropped matrix (resp. vector) obtained by restricting the columns (resp. components) of $\bs B$ (resp. $\bs u$) to those indexed in $\cl S \subset [D']$, $\bs B_{:,j}$ (or $\bs B_{j,:}$) are the $j^{\rm th}$ column (resp. row) of $\bs B$. The $\ell_2$-norm of vectors is $\|\bs u\| = (\bs u^* \bs u)^{1/2}$, while the Frobenius norm and scalar product of matrices are related by $\|\bs B\|_F = (\scp{\bs B}{ \bs B}_F)^{1/2} = (\tr \bs B^* \bs B)^{1/2}$. The $\ell_2$ (or Frobenius) unit ball in $\bb R^N$ (resp. $\bb R^{D \times D'}$) is denoted by $\bb B^N$ (resp. $\bb B^{D \times D'}_F \simeq \bb B^{DD'}$). The uniform distribution over $[-\frac{\delta}{2},\frac{\delta}{2}]$ is denoted $\cl U^{\bb R}_\delta$, and its complex counterpart is $\cl U^{\bb C}_\delta := \cl U^{\bb R}_\delta + \im \cl U^{\bb R}_\delta$.

\section{Radar System Model}
\label{sec:Radar-system}

In this section, we show how the 2D target localization information is \emph{linearly} encoded in the signals recorded by a radar system involving two antennas illustrated in Fig.~\ref{ULA}(left). 

\vspace{2mm}
\myp{Transmitted signal model:} The signal transmitted from an antenna located on the origin $(0,0)$ reads\hsq
\begin{equation*}
 \ts   s(t)=\sqrt{P_t}\,\exp\big(2\pi \im ( \int_{0}^t f_c(\xi) \ud\xi +\phi_0)\big),\hsq
\end{equation*}
where $P_t$ is the transmitted power, $f_c(t)$ is the transmitted frequency pattern, and $\phi_0$ is the initial phase of the oscillator. In this model, the frequency pattern $f_c(t)$ of a FMCW radar is characterized as a periodic saw-tooth function according to\hsq 
\begin{equation*}
 \ts   f_c(t)=f_0+ B\,(\frac{t}{T} \bmod 1),\hsq
\end{equation*}
where $f_0$ is the central frequency, $\bmod$ is the modulo operation, $T$ is the saw-tooth period, and $B$ is bandwidth spanned by the radar, \ie $f_c(t) \in [f_0, f_0 + B]$.

\vspace{2mm}
\myp{Digital Beam-forming reception model:} Let us first consider one static target located at range $R>0$ and angle ${\theta \in [-\pi,\pi]}$ from a receiving linear array comprised of two receiving antennas $\cl A_1$ and $\cl A_2$, located in $(0,0)$ and $(0, d)$, respectively (see Fig.~\ref{ULA}).
\begin{figure}[!t]
\centering
\includegraphics[width=0.45\columnwidth, height=5cm]{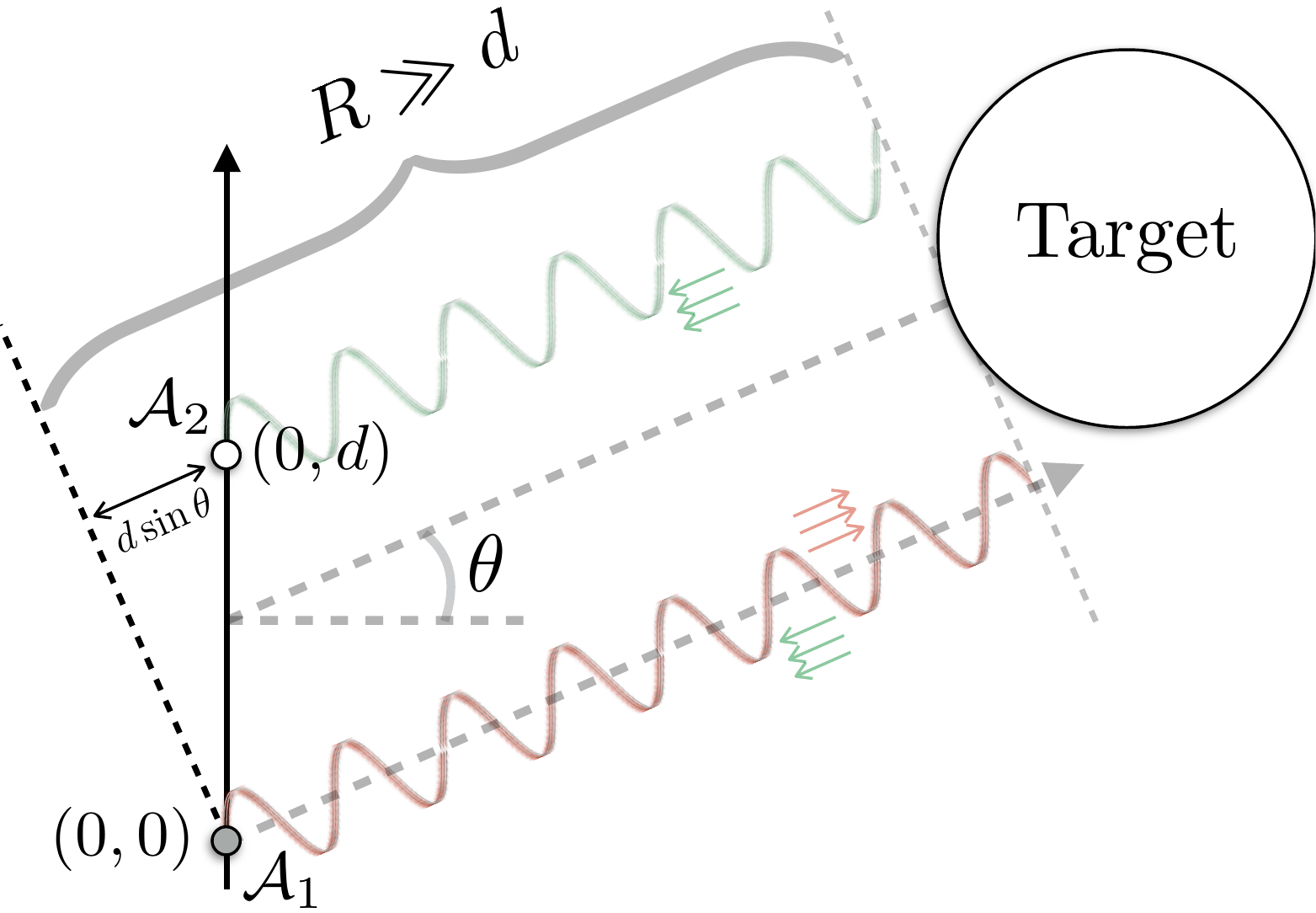} \hspace{1mm}
\quad\includegraphics[width=7cm, height=5cm]{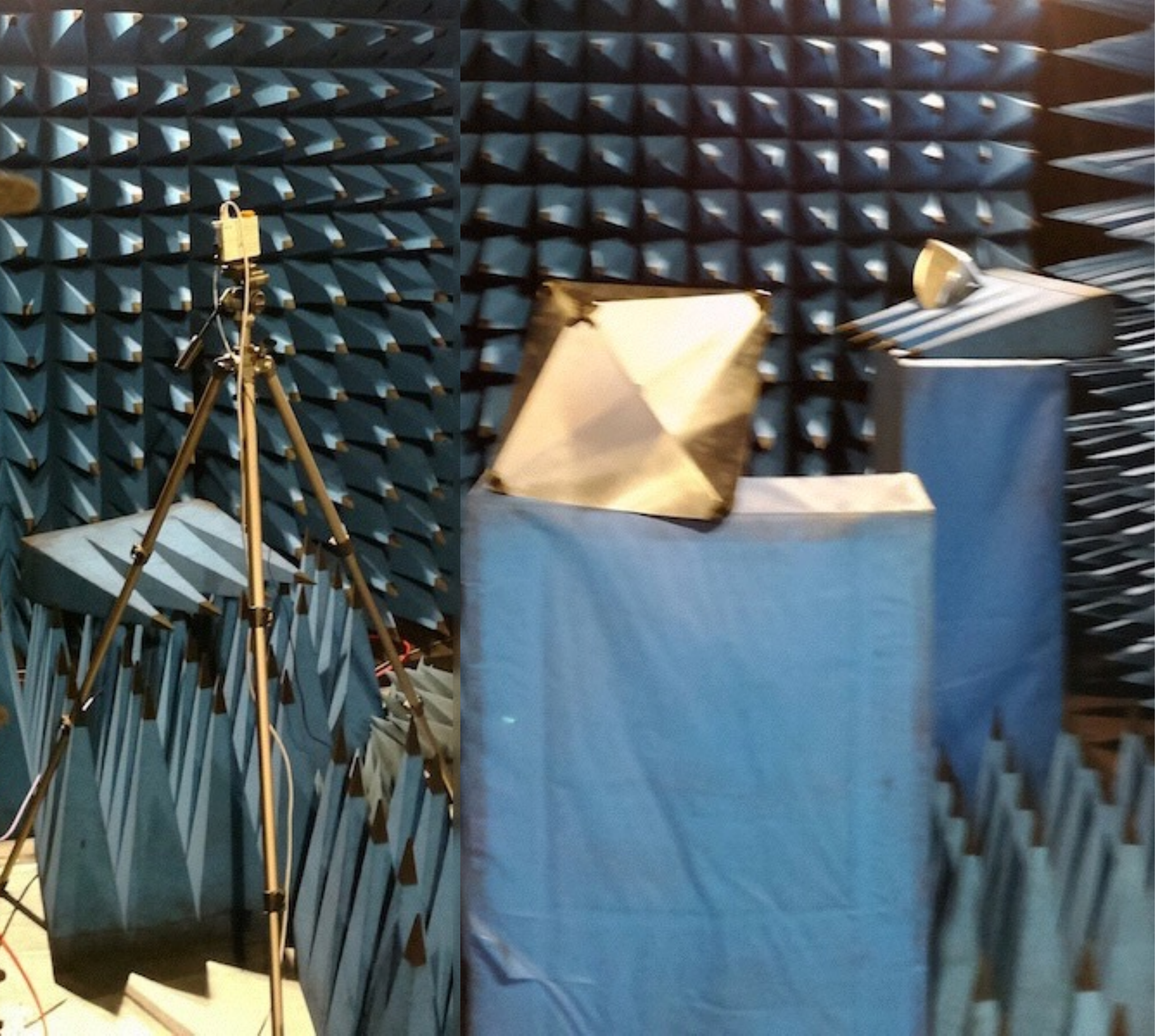}
\caption{\ninept (left) Illustration of the two antennas radar systeman array of receiving antennas. Experiment set-up with a FMCW radar on the left and two corner reflectors on the right (cropped image).\sq\sq\sq}
\label{ULA}
\end{figure}
The signal received on $\cl A_p$ ($p \in \{1, 2\}$) is:\hsq 
\begin{equation*}
    \gamma_p(t) = A\,s(t-\tau_p),\hsq
\end{equation*}
where $A$ is the complex received amplitude coefficient that depends on several parameters such as the range $R$ and the \textit{Radar Cross-Section} (RCS). Under the far-field approximation, the delay $\tau_p=c^{-1}( 2 R + p\,d\,\sin \theta)$ is the round-trip time between the transmitting antenna in $(0,0)$, the target and $\cl A_p$ ($p \in \{1,2\}$), with $c$ the speed of light.

After coherent demodulation of the FMCW radar signals, the acquisition model links the sampling time with the frequency that is being transmitted. The sampled signal can be seen as a measurement of the phase-shift at time $t$ of the transmitted frequency $f_c(t)$ depending on the target position. For a regular sampling at rate $M/T$, the sampled frequencies are $f_m := f_c(m \frac{T}{M})$, $1\leq m \leq M$, so that, at the $m^{\rm th}$ frequency,  $\cl A_p$ receives the signal:\hsq
\begin{align}
    \label{approxf0}
   \ts \Gamma_{mp}&= \ts x\, e^{-\im 2 \pi f_m \tau_p} \approx x\, e^{-\im 2 \pi f_m  \frac{2 R}{c}} e^{-\im 2 \pi f_0\frac{p d \sin \theta}{c}},\hsq
\end{align}
where $x$ is the received amplitude after the coherent demodulation. The approximation in \eqref{approxf0} is reasonable for K-band radars where $B \ll f_0$, \ie ${B=250{\rm MHz}}$ and ${f_0=24{\rm GHz}}$ respectively.

Comparing \eqref{approxf0} for $\cl A_1$ and $\cl A_2$ shows that the angle of arrival acts as a complex gain on $\cl A_2$. Furthermore, this paper considers a multi-target scenario using a purely additive model. This means that all the targets are in a direct line of sight from the radar, \ie there is no multi-path. For a scene with $K$ targets, recasting \eqref{approxf0} into a linear matrix sensing model and taking advantages of the phase relation between $\cl A_1$ and $\cl A_2$, the sensed signals $\bs \Gamma = \{\Gamma_{mp}\}_{mp} \in \bb C^{M \times 2}$ are\hsq
\begin{equation}
\label{eq:setting-signal}
\ts \bs \Gamma = [\bs \gamma_1, \bs \gamma_2] = \bs \Phi\, [\bs x, \bs G  \bs x],\hsq
\end{equation}
where $\bs x \in \mathbb{C}^{N}$ encodes the range profile, \ie $x_n \neq 0$ if there exists a target at range $R_n$, ${\|\bs x\|_0 := |\supp \bs x| \le K \ll N}$, ${\bs \Phi = \{ e^{-\im \frac{4\pi}{c} f_m R_n }\}_{mn} \in \bb C ^{M\times N}}$ is the \emph{range} measurement matrix, $\bs G =\diag(e^{\im \frac{2 \pi}{c} f_0 d \sin \theta_1},\,\cdots, e^{\im \frac{2 \pi}{c} f_0 d \sin \theta_N})$ with $\supp (\bs \theta)= \supp (\bs x)$, \ie $\bs G$ is the phase difference between the first and second receiving antennas. Therefore, the 2D-localization problem is tantamount to estimating the support $T$ of $\bs x$ from the sensing model \eqref{eq:setting-signal}, hence extracting the target ranges according to the discretization $\{R_n\}$. Comparing the phases of $\bs x_1$ and $\bs x_2$ on the index set $T$ then allows to deduce the angles $\bs \theta_T$. Interestingly, in this process, only the target ranges are discretized, \ie the angles are estimated from continuous phase differences. This, however, comes at a cost as this simplified two-antenna model does not allow the recovery of multiple targets located on the same range $R_n$.

Note that, in the absence of quantization, inverting \eqref{eq:setting-signal} can be solved using Maximum Likelihood \cite{mkay} if $M \geq N$, or using CS algorithms (\eg IHT \cite{BD2009} or CoSaMP \cite{NT2009}) if $M \leq N$. 

\section{Quantizing Radar observations}
\label{sec: QCS-model}

In this work, we propose to quantize the observations $\bs \Gamma$ achieved in the digital beamforming model \eqref{eq:setting-signal}. Our quantization procedure relies on a uniform scalar quantizer $\lambda \in \bb R \mapsto \cl Q(\lambda) = \delta \lfloor \frac{\lambda}{\delta}\rfloor + \frac{\delta}{2} \in \bb Z_\delta := \delta \bb Z + \frac{\delta}{2}$, with quantization width $\delta > 0$ applied entrywise onto vector or matrices, and separately onto the real and imaginary parts if these objects are complex.  

Our global objective is thus to estimate the localization of targets, as encoded in the matrix 
$\bs X = (\bs x_1, \bs x_2) = (\bs x, \bs G \bs x) \in  \bb C^{N\times 2}$, from the quantized observation model\hsq
\begin{equation}
  \label{eq:QRadar-system}
\bs Z = \Amap(\bs X) := \big(\Amap(\bs x_1), \Amap(\bs x_2)\big),\hsq
\end{equation}
with $\bs u \in \bb C^N \mapsto \Amap(\bs u) = \cl Q(\bs \Phi \bs u + \bs \xi) \in \bb Z^M_\delta + \im \bb Z^M_{\delta}$.  In $\Amap$, a uniform random \textit{dithering} $\bs \xi \in \bb C^{M}$, \ie $\xi_{m}  \sim_{\iid} \cl U^{\bb C}_\delta$ for all $m \in [M]$, is added to the quantizer input. For real sensing models, such a dithering attenuates the impact of the quantizer on the estimation of sparse/compressible signals in quantized CS \cite{PB2012,JC2017,XJ2018}. As will be clearer below, $\bs \xi$ also enables accurate estimation of $\bs X$.

As written in \eqref{eq:QRadar-system}, we can identify a low-complexity model for $\bs X$ in the case where only $K$ targets, with distinct ranges, are observed. We quickly see that\hsq 
\begin{equation*}
\ts  \bs X \in \Theta^K := \bigcup_{T \subset [N], |T|\leq K} \Theta_T,\hsq
\end{equation*}
with $\Theta_T := \{\bs U \in \bb C^{N \times 2}: \supp(\bs U_{:,1}) = \supp(\bs U_{:,2}) \subset T\}$, which is a union of ${N \choose K}$ $K$-dimensional subspaces. Note that, according to \eqref{eq:setting-signal}, since $|G_{jj}|=1$ we could further impose $|X_{j1}| = |X_{j2}|$ for all $j\in [N]$. However, this leads to an hardly integrable non-convex constraint on the domain of $\bs X$.  

In summary, the acquisition model \eqref{eq:QRadar-system} can thus be seen as a quantized embedding of $\Theta^K \ni \bs X$ into $\bs{\sf A}(\Theta^K) \subset \bb Z_\delta^M + \im \bb Z_\delta^M$ via the quantized mapping $\Amap$ \cite{XJ2018}. In addition, the model~\eqref{eq:QRadar-system} leads to 1-bit measurements with values $\{\pm\delta/2\}$ when $\delta$ is sufficiently large, \ie $\delta >  2 \max(\|\bs \Phi \bs x_1\|_\infty, \|\bs \Phi \bs x_2\|_\infty)$.

\section{2D Target Localization in Quantized Radar}
\label{sec:2d-targ-local-qradar}

Despite the quantization, the sensing model \eqref{eq:QRadar-system} still enables target localization. We adopt here a simple method, the projected back projection (PBP) proposed in \cite{XJ2018}, for which the estimation error provably decays when the number of observations $M$ increases. The PBP estimate is defined from
\begin{equation}
  \label{eq:PBP}
  \hat{\bs X} = \cl P_{\Theta^K}(\tinv{M} \bs \Phi^* \bs Z), 
\end{equation}
with the projector $\cl P_{\Theta^K}(\bs U) := \arg\min_{\bs V \in \Theta^K} \|\bs U - \bs V\|^2_F$. In words, the estimate $\hat{\bs X}$ is achieved by first \emph{back projecting} $\bs Z$ in the signal domain thanks to the adjoint sensing $\bs \Phi^*$, and then taking the closest point in $\Theta^K$ to $\tinv{M} \bs \Phi^* \bs Z$ from the projector~$\cl P_{\Theta^K}$. 

Interestingly, for any $\bs U \in \bb C^{N \times 2}$, $\bs V = \cl P_{\Theta^K}(\bs U)$ is easily computed.  Denoting by $\cl H_K(\bs v)$ the hard thresholding operator setting all but the $K$ largest components (in magnitude) of $\bs v \in \bs C^N$ to zero, we first form $T = \supp \cl H_K(\bs u)$ with $\bs u = \cl S(\bs U) \in \bb R^N_+$ and $\cl S(\bs U)_n = (|U_{n,1}|^2 + |U_{n,2}|^2)^{1/2}$ for all $n\in [N]$, and then, for $p \in \{1,2\}$, $V_{np}$ equals to $U_{np}$ if $n \in T$, and to 0 otherwise.  This provides clearly $\|\cl P_{\Theta^K}(\bs U) - \bs U\|^2_F = \|\cl H_K(\cl S(\bs U)) - \cl S(\bs U)\|^2 \leq \|\tilde{\bs u} - \cl S(\bs U)\|^2$ for any $\tilde{\bs u} \in \bb R^n_+$ such that $|\supp \tilde{\bs u}|\leq K$.  Since for any $\bs U' \in \Theta_K$, $\tilde{\bs u}:=\cl S(\bs U') \in \bb R^n_+$, we thus have $\|\cl P_{\Theta^K}(\bs U) - \bs U\|^2_F \leq \|\cl S(\bs U') - \cl S(\bs U)\|^2 = \|\bs U' - \bs U\|_F^2$ as required from the definition of $\cl P_{\Theta^K}$.

Given the estimate $\hat{\bs X}$ in \eqref{eq:PBP}, the range profile $\hat{T}$ is simply obtained as $\hat{T} = \supp\big(\cl S(\hat{\bs X})\big)$, so that targets are localized in the polar coordinates $(R_n,\theta_n)$ for all $n \in \hat T$ with $\theta_n = \arcsin\big(\frac{c}{2\pi f_0 d} \angle(\hat x_2[n]^* \hat x_1[n])\big)$. 

We now establish how the estimation error $\|\hat{\bs X} - \bs X\|_F$ of~\eqref{eq:PBP} can be bounded with high probability. This is important to ensure the quality of the estimated target coordinates. To this end, given $0<\epsilon<1$, we first assume that our radar sensing matrix $\sinv{\sqrt{M}}\bs \Phi$ respects the restricted isometry property over the set $\bar{\Sigma}^N_K :=\{\bs u \in \bb C^N: |\supp \bs u| \leq K\}$ of complex $K$-sparse signals, in short $\sinv{\sqrt{M}}\bs \Phi \in$ RIP$(\bar{\Sigma}^N_K,\epsilon)$, \ie for all $\bs u \in \bar{\Sigma}^N_K$, 
$$
\ts (1 - \epsilon) \|\bs u\|^2 \leq \inv{M} \|\bs \Phi \bs u\|^2 \leq (1 + \epsilon) \|\bs u\|^2.
$$

Many random matrix constructions have been proved to respect the RIP with high probability (\whp\footnote{``\whp'' means with probability exceeding $1-Ce^{-c\epsilon^2M}$ for $C,c>0$.}) \cite{BDDW08,FR2013,candes2006stable}. Given the discrete Fourier matrix $\bs F\in \bb C^{N \times N}$, the selection matrix $\bs R_\Omega$ such that $\bs R_\Omega \bs u = \bs u_{\Omega}$, and provided $M \gtrsim \epsilon^{-2} K\, (\ln K)^2 \ln N$, if $\Omega \subset [N]$ has cardinality $M$ and is picked uniformly at random among the ${N \choose M}$ $M$-length subsets of $[N]$, then $\bs \Phi =\sqrt{N} \bs R_\Omega \bs F$ respects \whp the RIP$(\bar{\Sigma}^N_K,\epsilon)$ \cite{FR2013,R2010}. Therefore, up to a random sub-sampling of the $M$ frequencies $\{f_m\}$, the radar sensing matrix $\bs \Phi$ follows a similar construction. 

Second, given ${\nu > 0}$, we assume that $\Amap$ satisfies the (complex) limited projection distortion over $\bar{\Sigma}^N_K$, or $\Amap \in$ LPD$(\bar{\Sigma}^N_K, \bs \Phi, \nu)$, \ie  
\begin{equation}
\label{eq:LPD}
 \ts \sinv{M} |\scp{\Amap(\bs w)}{\bs \Phi \bs v} - \scp{\bs \Phi \bs w}{\bs \Phi \bs v}| \leq \nu,\ \forall \bs w,\bs v \in \bar{\Sigma}^N \cap \bb B^{N}. 
 \end{equation}

Thanks to these two conditions, we get the following guarantee on $\hat{\bs X}$. 
\begin{Propn}
\label{prop:PBP-guarantees}
Given $\epsilon, \mu > 0$, if ${\sinv{\sqrt M}\bs \Phi \in {\rm RIP}(\bar{\Sigma}^N_{2K}, \epsilon)}$ and~${\bs A \in {\rm LPD}(\bar{\Sigma}^N_{2K}, \bs \Phi, \nu)}$, then, for all $\bs X \in \Theta_K$ the PBP estimate \eqref{eq:PBP} satisfies $\|\hat{\bs X} - \bs X\|_F \leq 2(\epsilon + 2\nu)$.     
\end{Propn}
\begin{proof}
If $\sinv{\sqrt M} \bs \Phi \in {\rm RIP}(\bar{\Sigma}^N_{2K}, \epsilon)$, then $\sinv{\sqrt M} \bs \Phi \in {\rm RIP}(\Theta^{2K}, \epsilon)$ with respect to the Frobenius norm since $|\sinv{M}\|\bs \Phi \bs U\|_F^2 - \|\bs U\|_F^2| \leq |\sinv{M}\|\bs \Phi \bs u_1\|^2 - \|\bs u_1\|^2| + |\sinv{M}\|\bs \Phi \bs u_2\|^2 - \|\bs u_1\|^2| \leq \epsilon (\|\bs u_1\|^2 + \|\bs u_2\|^2) = \epsilon \|\bs U\|_F^2$,  for all $\bs U = (\bs u_1, \bs u_2) \in \Theta^{2K} \subset (\bar{\Sigma}^N_{2K}, \bar{\Sigma}^N_{2K})$.   
Moreover, extending the LPD \eqref{eq:LPD} to matrices with the Frobenius scalar product, if $\Amap \in {\rm LPD}(\bar{\Sigma}^N_{2K}, \bs \Phi, \nu)$, then the matrix map $\Amap \in {\rm LPD}(\Theta^{2K}, \bs \Phi, 2\nu)$
since $\scp{\Amap(\bs W)}{\bs \Phi \bs V}_F = \scp{\Amap(\bs w_1)}{\bs \Phi \bs v_1} + 
\scp{\Amap(\bs w_2)}{\bs \Phi \bs v_2}$ for any $\bs W = (\bs w_1,\bs w_2), \bs V = (\bs v_1, \bs v_2) \in \Theta^K$, and similarly for $\scp{\bs \Phi \bs W}{\bs \Phi \bs V}_F$. 
The rest of the proof is a quick extension of \cite[Thm. 4.1]{XJ2018} to complex $N\times 2$ matrices belonging to the union of low-dimensional subspaces $\Theta^K$.    
\end{proof}

The next proposition (proved in Appendix) determines when $\Amap$ respects the LPD, as required by Prop.~\ref{prop:PBP-guarantees}.
\begin{Propn}
\label{prop:lpd-existence}
Given $\epsilon > 0$, if $\sinv{\sqrt M}\bs \Phi \in {\rm RIP}(\bar{\Sigma}^N_{2K},\epsilon)$ and if 
$M \geq C \epsilon^{-2} K \ln(\frac{N}{K})\ln(1 + \frac{c}{\epsilon^3})$, then, \whp, $\Amap \in {\rm LPD}(\bar{\Sigma}^N_K, \bs \Phi, 4 \epsilon (1+\delta))$.   
\end{Propn}

\noindent In other words, inverting the role of $M$ and $\epsilon$ in the requirement of Prop.~\ref{prop:lpd-existence} and assuming $\sinv{\sqrt M}\bs \Phi$ is RIP, up to some log factors, Prop.~\ref{prop:PBP-guarantees} shows that, \whp, the estimation error of PBP decays like $\|\hat{\bs X} - \bs X\|_F = O\big((1+\delta)\sqrt{{K}/{M}}\big)$ when $M$ increases.

\section{Numerical Results}
\label{sec:num-res}

In this section, Monte Carlo simulations are performed for different sparsity level to assess the accuracy of the proposed scheme for a variety of targets' positions.

\vspace{2mm}
\myp{A. Parameters and metrics:}  We simulate the working mode of a noiseless K-Band radar, \ie giving $f_0 = 24{\rm Ghz}$ and a bandwidth of $B=250 {\rm Mhz}$. The spacing $d$ between the two antennas is defined as half a wavelength, \ie $d=\frac{c}{2 f_0}$, allowing for angular estimation in $[-\frac{\pi}{2}, \frac{\pi}{2}]$. In all our simulations, we set the number of ranges $N$ to $256$, giving a range limit of $R_{\max}=153.6$m and a range resolution of $0.6$m. We test $320\,000 \times K$ Monte Carlo runs, where $K$ is the considered sparsity, the targets' localization are picked uniformly at random in a $40\times 40$ discretized polar domain $(R, \theta) \in [0,R_{\max}] \times [-\pi/2, \pi/2]$. In this domain, all targets receive uniformly random phases in $[0, 2\pi]$, the strongest target being set to a unit amplitude and the weaker ones having uniformly distributed amplitudes in $[0, 1]$. In order to focus on bit-rate reduction in radar processing, a total budget of 512 bits per channel is fixed for each acquisition with $M$ measurements, \ie giving $M=512$ measurements for 1-bit measurement quantization (\ie 2 complete FMCW saw-tooths), or $M=16$ for 32-bit measurements. Our regime thus leads to a bit-rate reduction of $93.75\%$ compared to a full acquisition with $M=256$ for 32-bit measurements. The quality of the position estimation is simply measured as $\min_k |R e^{\im \theta}-\hat R_k e^{\im \hat \theta_k}|$, \ie the distance between the true target location and the closest estimated targets in $(\hat{ \bs R},\hat{ \bs \theta})$. This quality measure is then averaged over runs which have the same position $(R,\theta)$. These results are reported in a 2D polar graph (Fig.~\ref{K1}, Fig.~\ref{fig:K2_all}, Fig.~\ref{fig:K2_sep}). 
\begin{figure}[!t] 
    \centering
    \begin{subfigure}[b]{0.25\textwidth}
        \includegraphics[height=2.8cm]{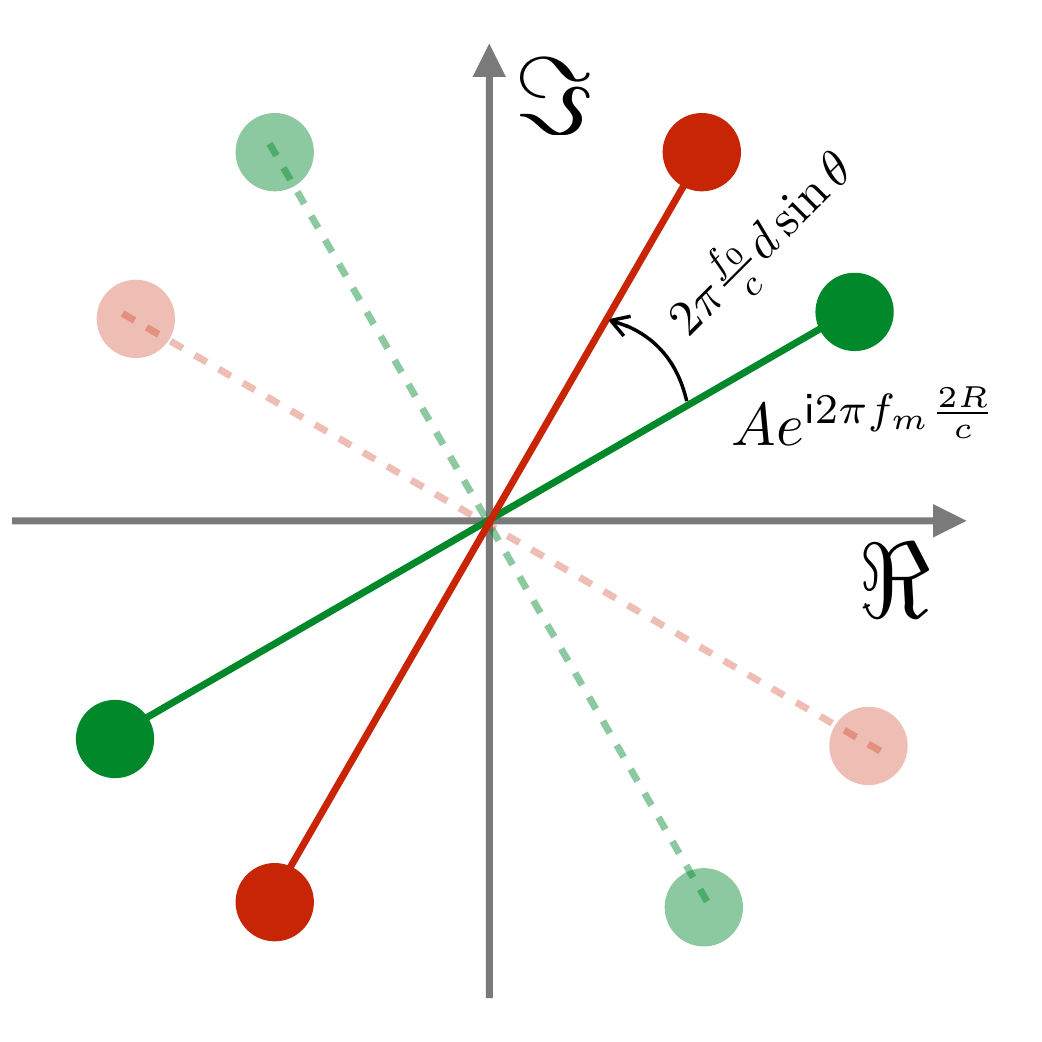}
        \caption{}
        \label{fig:CE}
    \end{subfigure}
    \begin{subfigure}[b]{0.20\textwidth}
        \includegraphics[height=2.8cm]{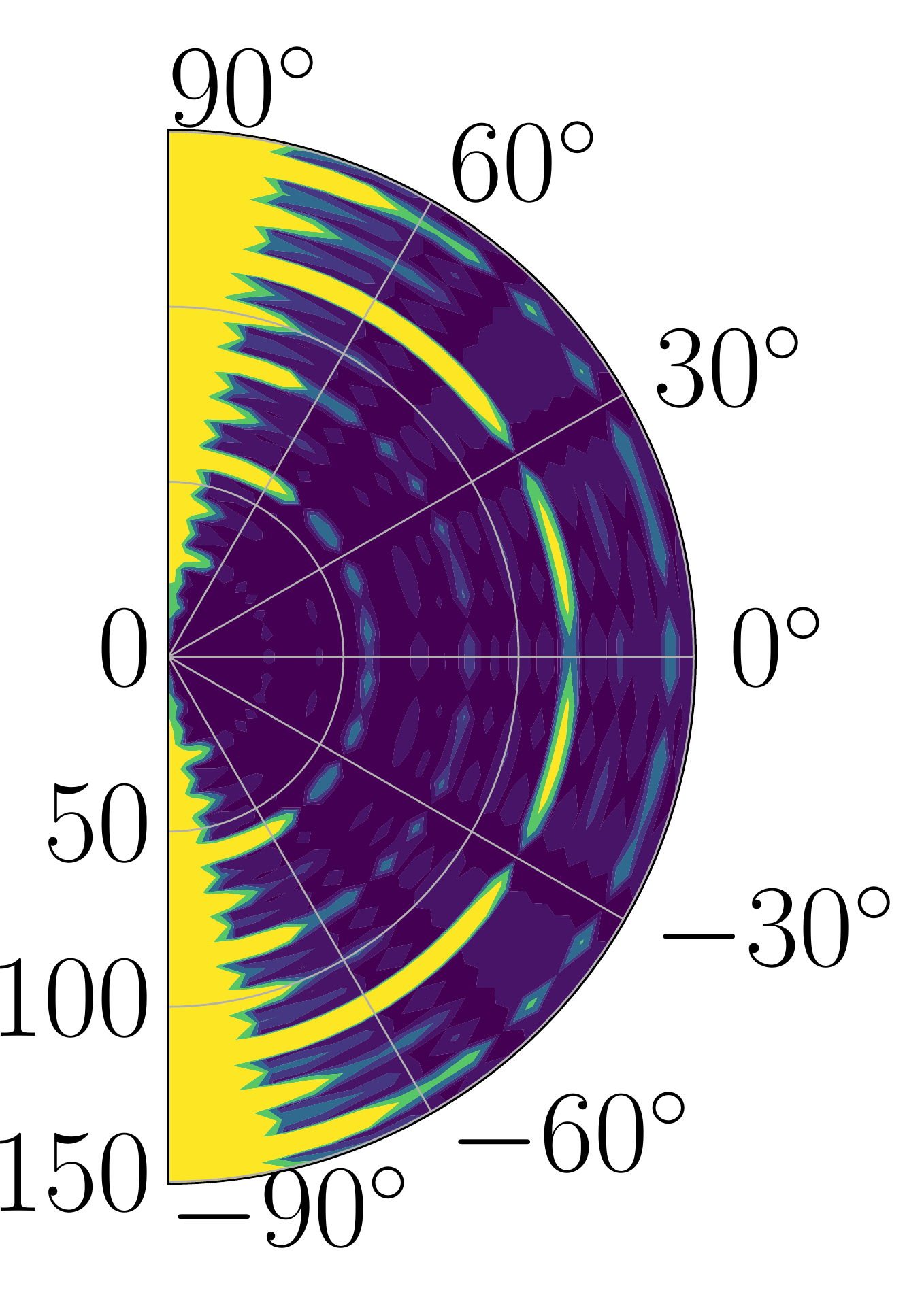}
        \caption{}
        \label{fig:1ND}
    \end{subfigure}
    \begin{subfigure}[b]{0.2\textwidth}
        \includegraphics[height=2.8cm]{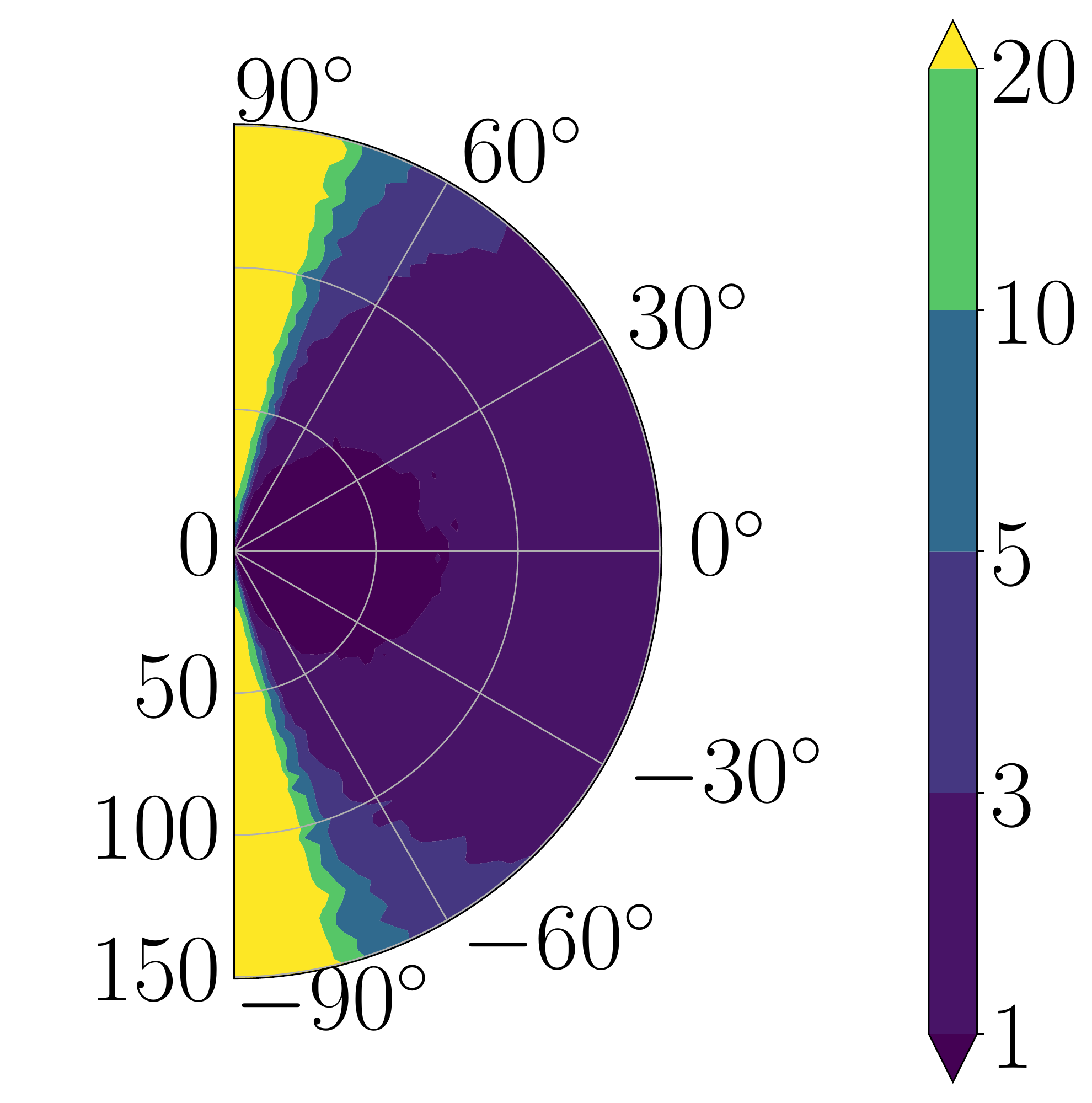}
        \caption{}
        \label{fig:1D}
    \end{subfigure}\vspace{-4mm}
      \caption{\ninept (a) Example of a possible 2D target localization ambiguity. (b) and (c), positions error in meters for Monte Carlo simulations with one target and $M=512$, for 1-bit non-dithered and 1-bit dithered quantization scheme, respectively.\sq}\label{fig:K1_all}
    \label{K1}
\end{figure}

\medskip
\myp{B. Simulations for a single target scenario:} In this first simulation we test the 2D-localization of a single target (\ie $\bs X \in \Theta^1 $) for the dithered and non-dithered scheme. 
In Fig.~\ref{fig:1ND}, the non-dithered scheme exhibits systematic artifacts in the estimation quality. Indeed, in the context of radar localization of one target, ambiguities appear when full resolution signals from different receiving antennas once quantized are the same. 
One obvious possibility is when the angle of arrival $\theta$ is so small that $ \bs G \approx {\rm \bf Id}$ which means that $\cl Q( \bs \Phi \bs x) = \cl Q( \bs \Phi  \bs G \bs x)$. Another possibility appears for targets with large angles of arrival at certain ranges, as seen in the artifact pattern in Fig.~\ref{fig:1ND}. Certain ranges induce a strong repetition between quantized measurements, as depicted for illustration in Fig.~\ref{fig:CE} for a range of $\tinv{4}R_{\max}$. The quantized signals being identical, the estimated angle is $0^{\circ}$ regardless of the actual angle. 

The dithered scheme  in  Fig.~\ref{fig:1D} exhibits good performances on a wide range of positions and the effect of the dithering is clearly visible by the absence of artifacts. The drop of performances in Fig. \ref{fig:1ND} and Fig.~\ref{fig:1D} at $\pm 72^{\circ}$ degrees is related to the sensitivity of the $\arcsin$ function.

\vspace{2mm}
\myp{C. Simulations for the 2 targets scenario:} Fig.~\ref{fig:K2_all} shows the performances of the schemes for 2 targets. When compared to the non-dithered  (Fig.~\ref{fig:K2ND}) or full resolution (Fig.~\ref{fig:K232}), the dithered strategy in Fig.~\ref{fig:K2D} surpasses the others constrained to the same bit rate.  
 \begin{figure}[!t]
    \centering
    \begin{subfigure}[b]{0.20\textwidth}
        \includegraphics[height=2.8cm]{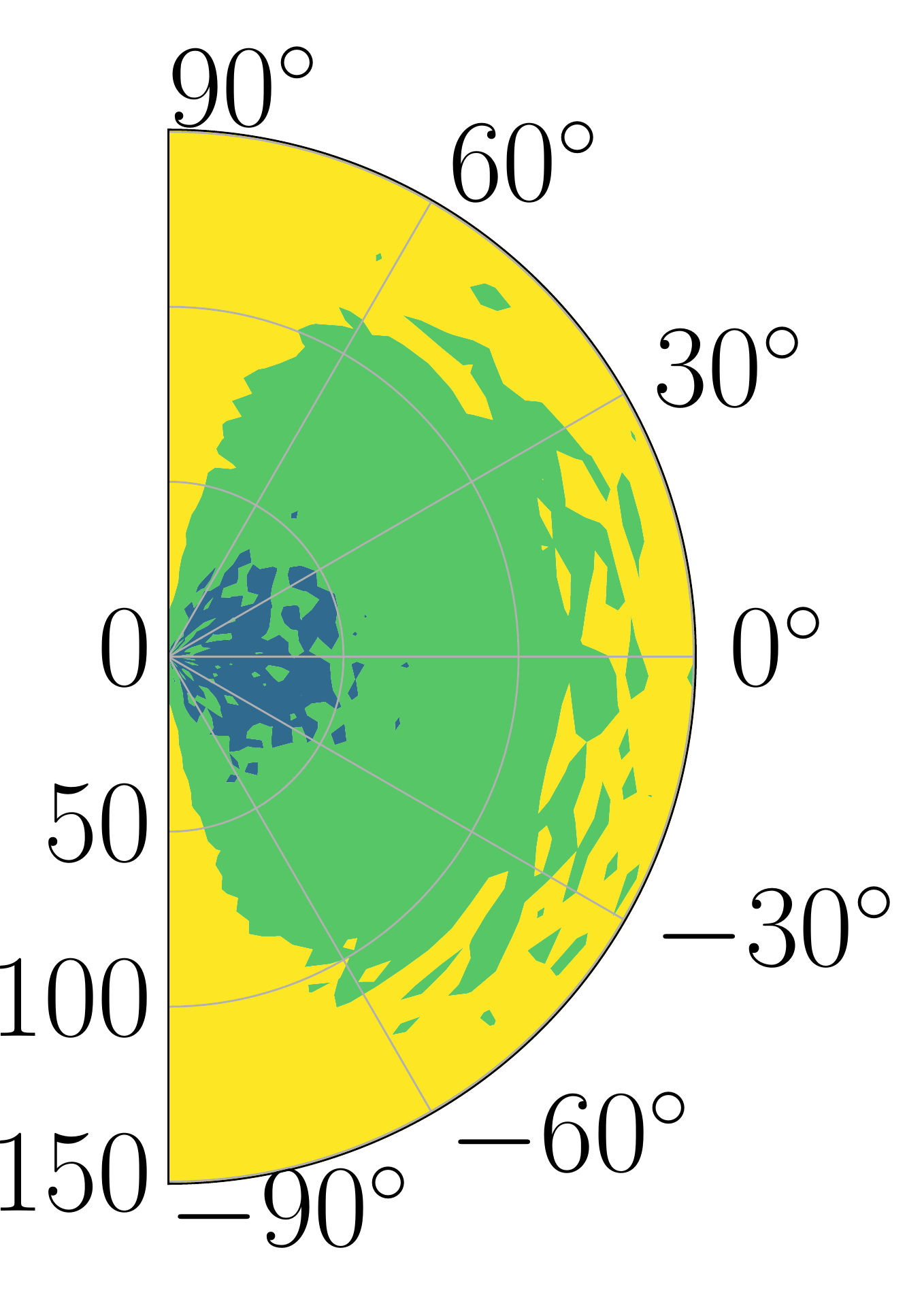}
        \caption{}
        \label{fig:K2ND}
    \end{subfigure}
    \begin{subfigure}[b]{0.20\textwidth}
        \includegraphics[height=2.8cm]{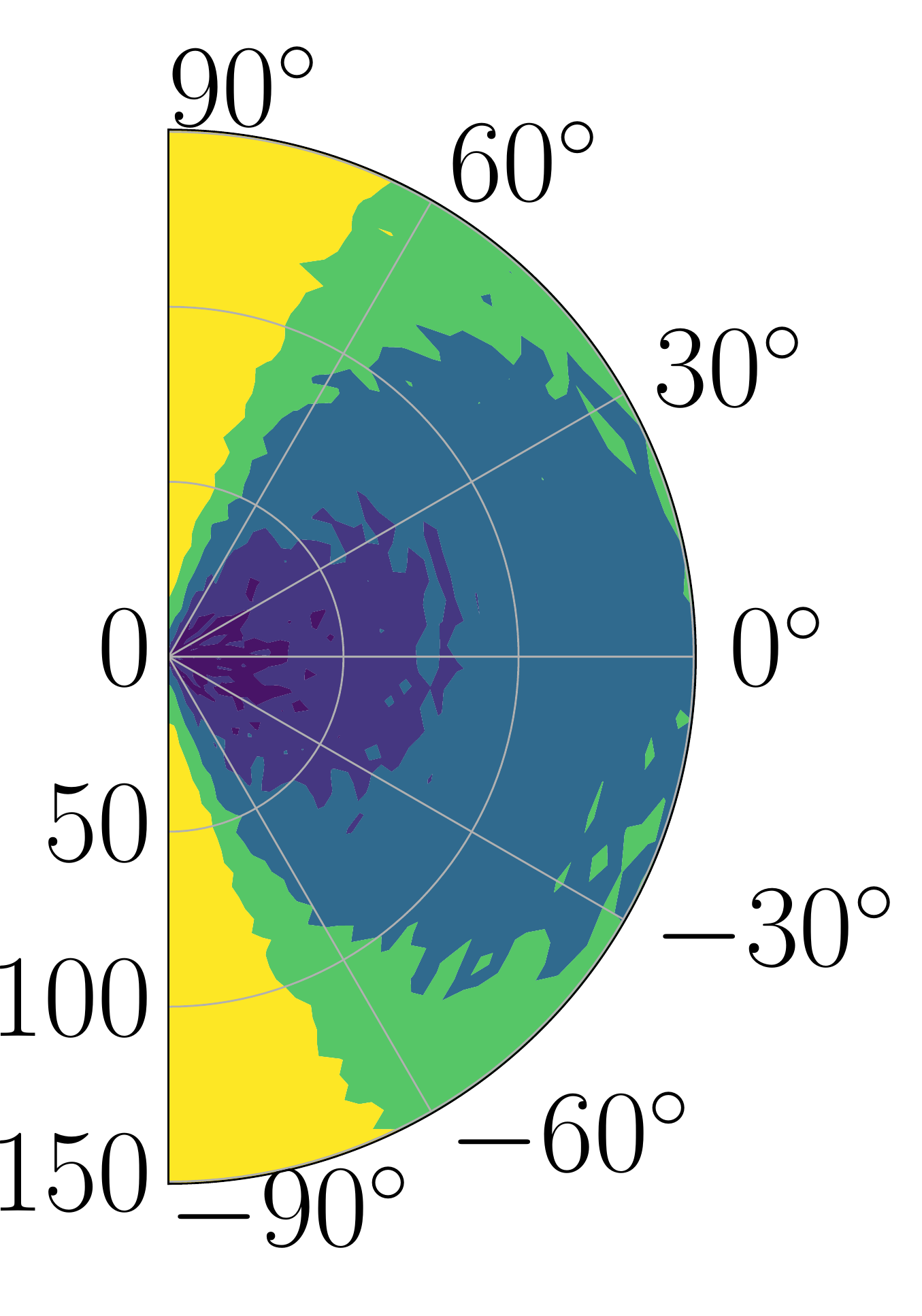}
        \caption{}
         \label{fig:K2D}
    \end{subfigure}
    \begin{subfigure}[b]{0.20\textwidth}
        \includegraphics[height=2.8cm]{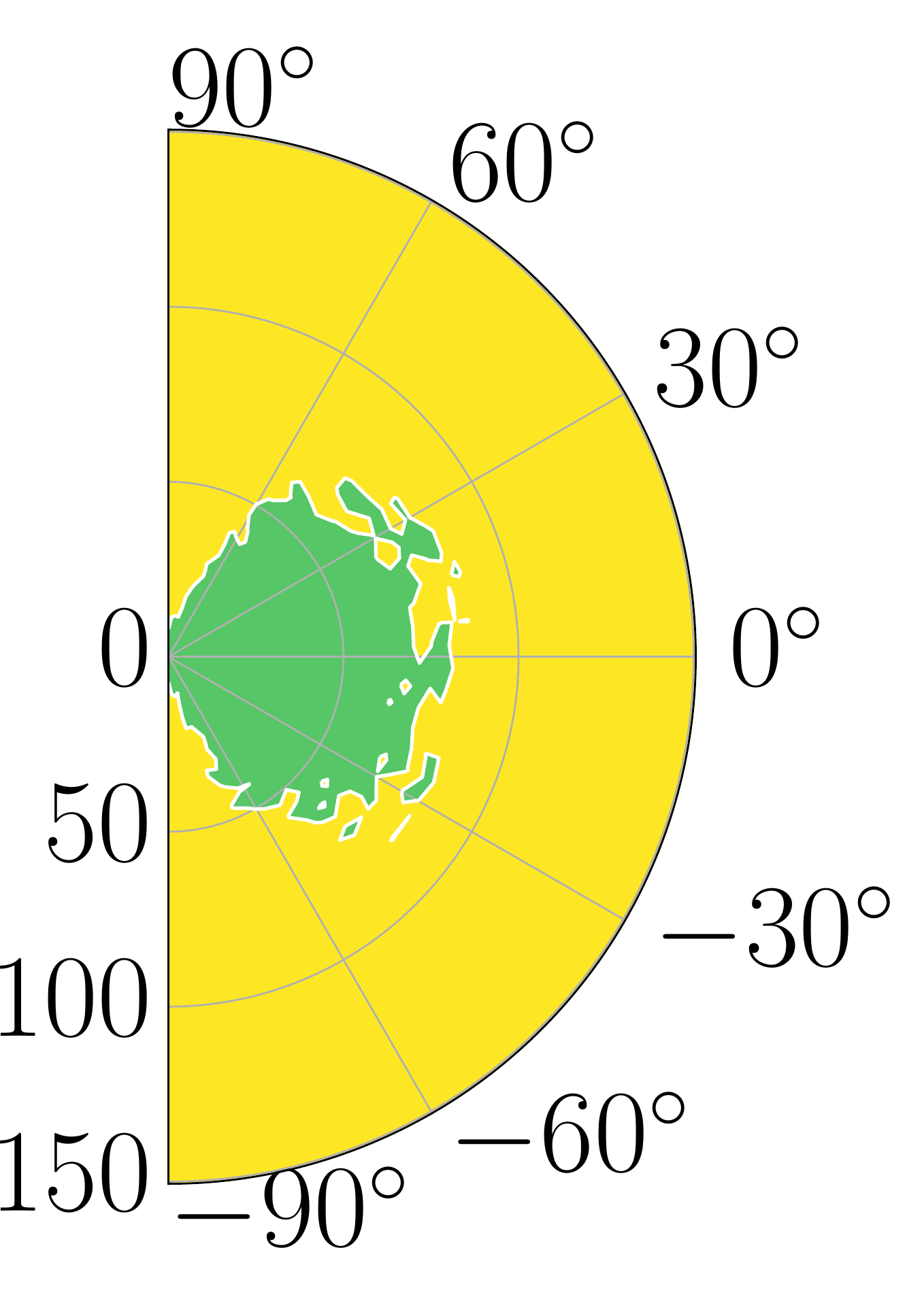}
        \caption{}
        \label{fig:K232}
    \end{subfigure}
    \begin{subfigure}[b]{0.2\textwidth}
        \includegraphics[height=2.8cm]{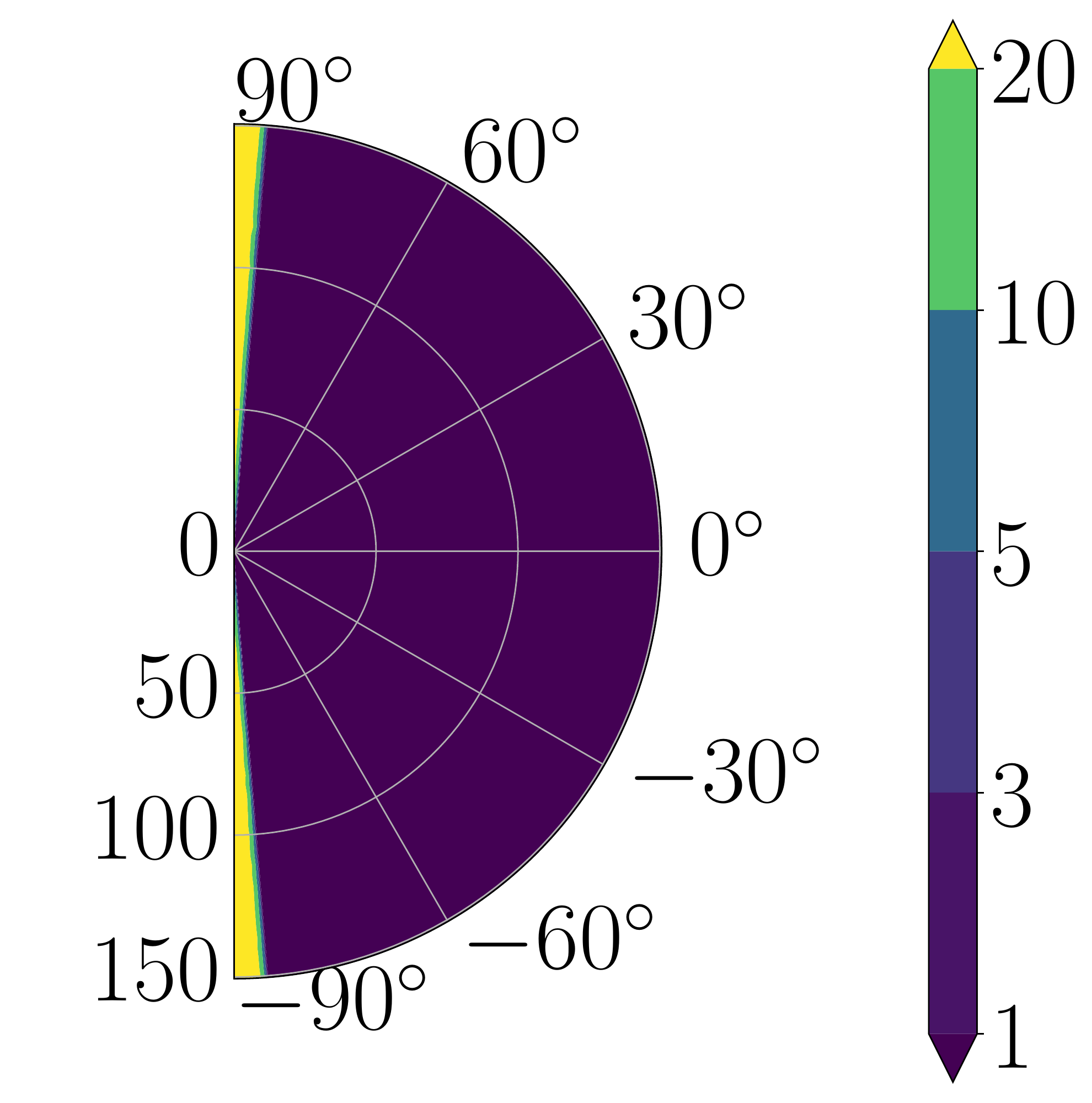}
        \caption{}
        \label{fig:K232F}
    \end{subfigure}\vspace{-4mm}
    \caption{\ninept Positions error in meters for Monte Carlo simulations with two targets; (a) $M=512$, 1-bit non-dithered quantization; (b) $M=512$, 1-bit dithered quantization; (c) $M=16$, 32-bit non-dithered; (d) $M=256$, 32-bit full measurements.\sq\sq}\label{fig:K2_all}
\end{figure}
Comparing Fig.~\ref{fig:1D} and Fig.~\ref{fig:K2D}, a drop of performances can be seen from the increase in sparsity. This is consistent with the results in Sec.~\ref{sec: QCS-model}, where we showed that the bound on the error of PBP grows as the sparsity increases. Moreover, in the absence of dithering in our quantized radar scheme, extremely sparse signals can lead to ambiguous estimations. The complete formulations and study of these situations are the subject of a future publication. In Fig.~\ref{fig:K2_sep}, the strongest and weakest estimated targets are separated to study their respective accuracy. For the non-dithered scheme in Fig.~\ref{fig:11ND} the strongest target still exhibits artifacts whereas the weakest (Fig.~\ref{fig:12ND}) is consistently wrongly estimated. The dithering reduces partly these situations and offers better performances for both targets in Fig.~\ref{fig:11D} and Fig.~\ref{fig:12D}.
While the accuracy of the second target for the dithered scheme is impressive when compared to the non-dithered one, it is far from what can be achieved in Fig.~\ref{fig:K232F} for the full measurements and resolution approach. This paper is one of the first venture into radar localization using 1-bit dithered scheme that is, furthermore, constrained to a specifically low bit-rate. In the future, better reconstruction qualities could be obtained by replacing PBP with other algorithms explicitly using the dithering to reach consistent signal estimates~\cite{DJR2017,BFNPW2017}.

\begin{figure}[!t]
    \centering
    \begin{subfigure}[b]{0.20\textwidth}
        \includegraphics[height=2.8cm]{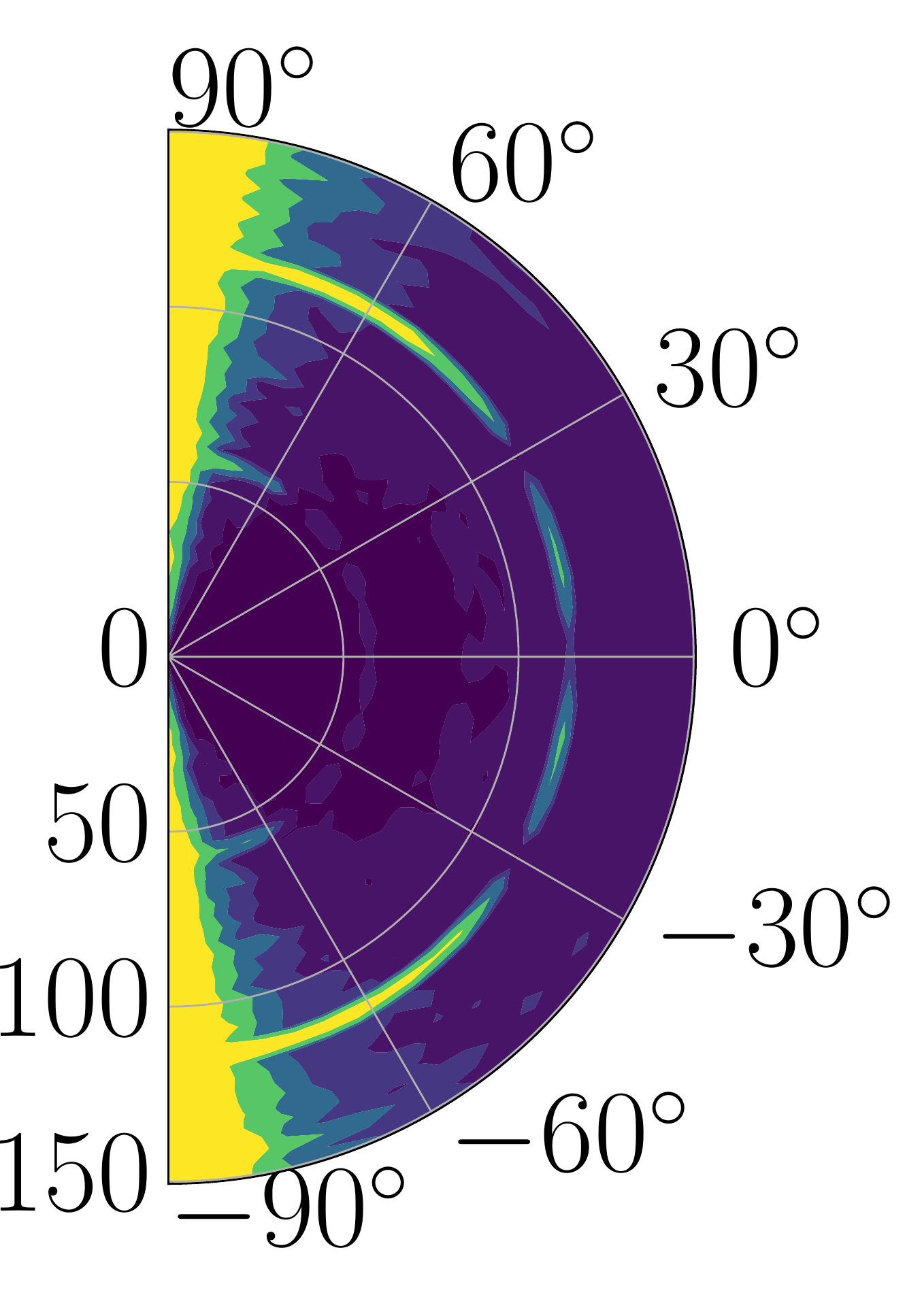}
        \caption{}
        \label{fig:11ND}
    \end{subfigure}
    \begin{subfigure}[b]{0.20\textwidth}
        \includegraphics[height=2.8cm]{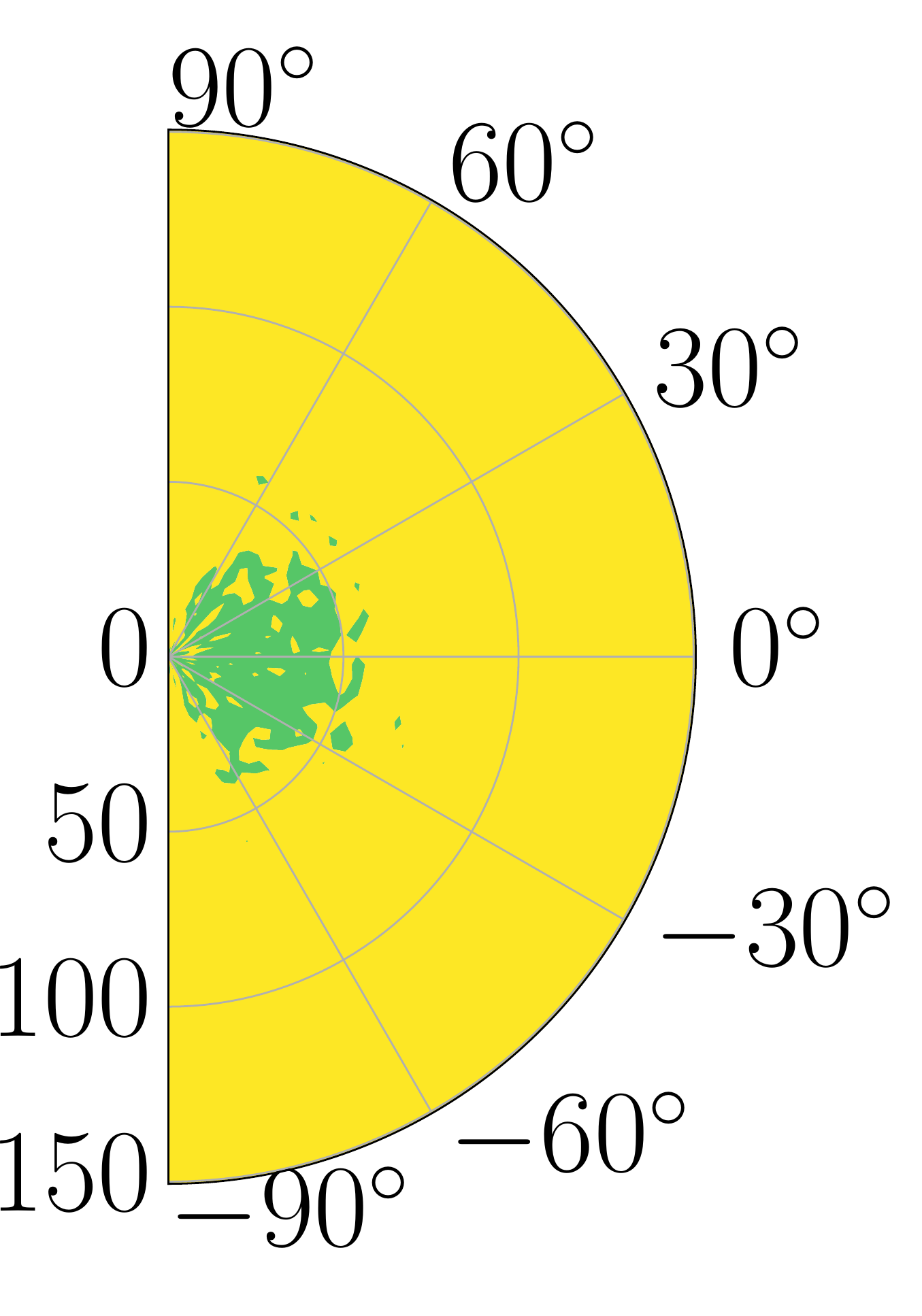}
        \caption{}
        \label{fig:12ND}
    \end{subfigure}
    \begin{subfigure}[b]{0.20\textwidth}
        \includegraphics[height=2.8cm]{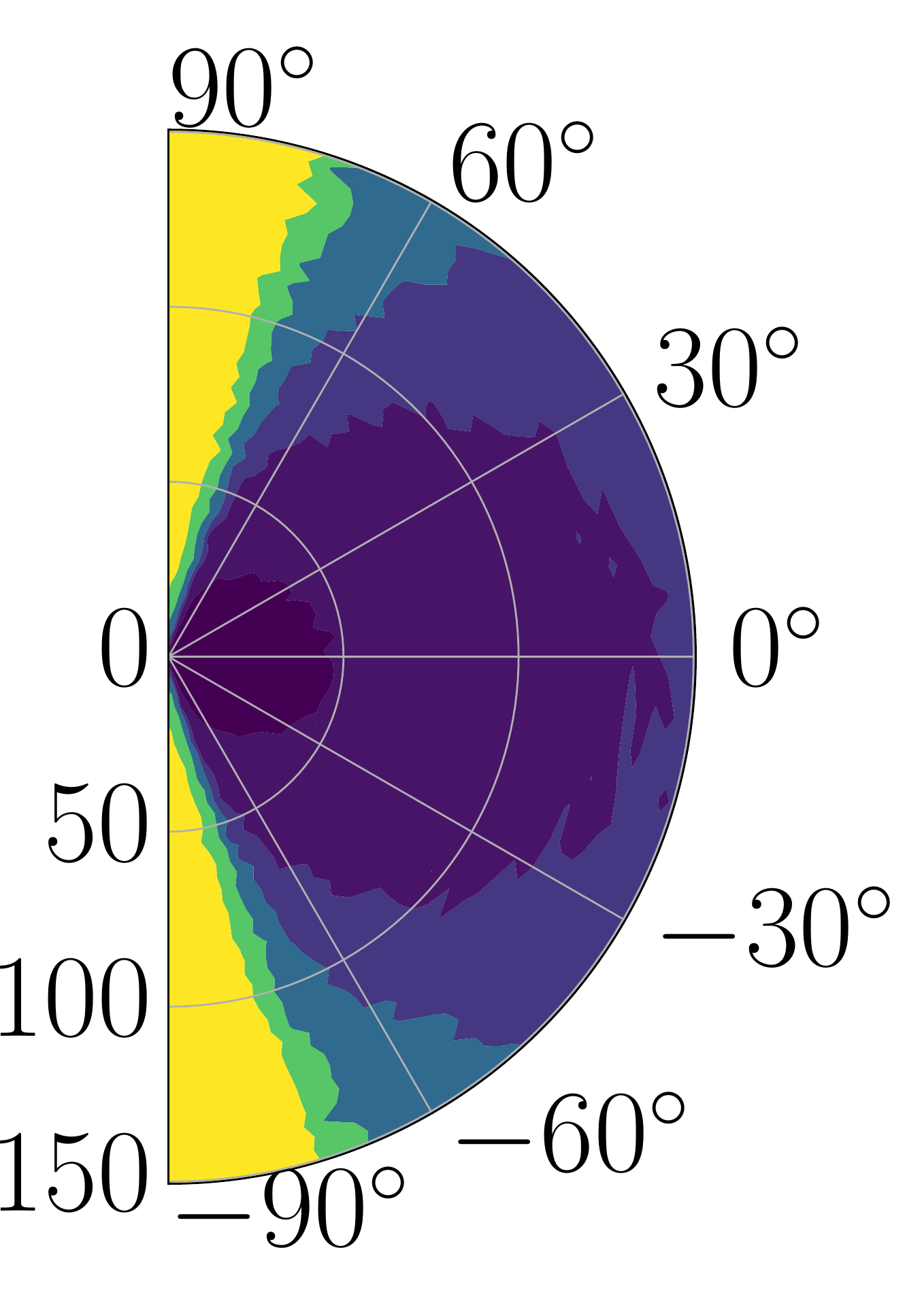}
        \caption{}
        \label{fig:11D}
    \end{subfigure}
    \begin{subfigure}[b]{0.2\textwidth}
        \includegraphics[height=2.8cm]{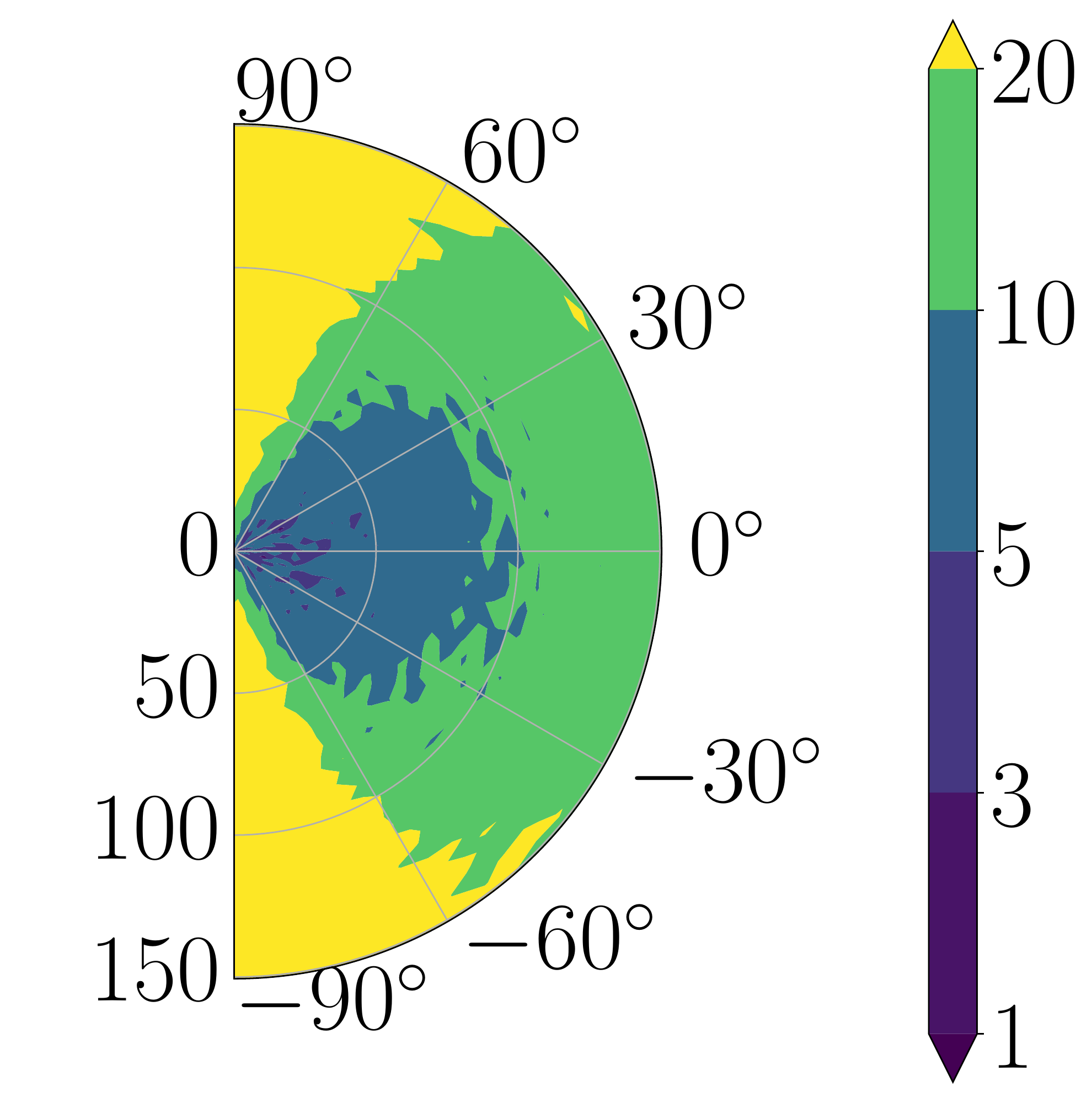}
        \caption{}
        \label{fig:12D}
    \end{subfigure}\vspace{-4mm}
  \caption{\ninept Positions error in meters for Monte Carlo simulations with two targets, $M=512$ and 1-bit quantization; in (a) and (b), strongest and weakest target for the non-dithered scheme, respectively; in (c) and (d), strongest and weakest target for the dithered scheme, respectively.\sq}\label{fig:K2_sep}
\end{figure}
\hsq
\section{Real Measurements}
\label{sec:real-experiment}

The study of the proposed scheme is now extended to real measurements to test the model and reconstruction algorithm against noise and nonidealities from the environment, the targets and the radar. 
The measurements are performed in an anechoic chamber where two targets are located in front of a commercial radar product \cite{KMD2} at different ranges and angles. The radar parameters, \eg $B$ and $N$, mirror the ones used in the simulations. The proposed scheme and the developed theory only considers the noiseless case. Practical measurements with real radar sensors are, however, inherently corrupted by noise. To this end, the reconstruction is studied with different levels of added dithering to assess the impact of the already present noise on the quantization. Fig.~\ref{fig:alpha} shows the mean position errors for the two targets for a weighted dithering $\tilde{\bs \xi}=\alpha \bs \xi$, with $\alpha \in [0,1]$, where $\bs \xi$ is the dithering defined in Sec.~\ref{sec: QCS-model}. Fig.~\ref{fig:alpha} shows that a certain amount of dithering is required to achieve good performance but also that adding a full dithering (\ie $\alpha=1$) is not the \textit{by default } optimum.
\begin{figure}[!t]
    \centering
    \begin{subfigure}[b]{0.4\textwidth}
        \includegraphics[width=\textwidth]{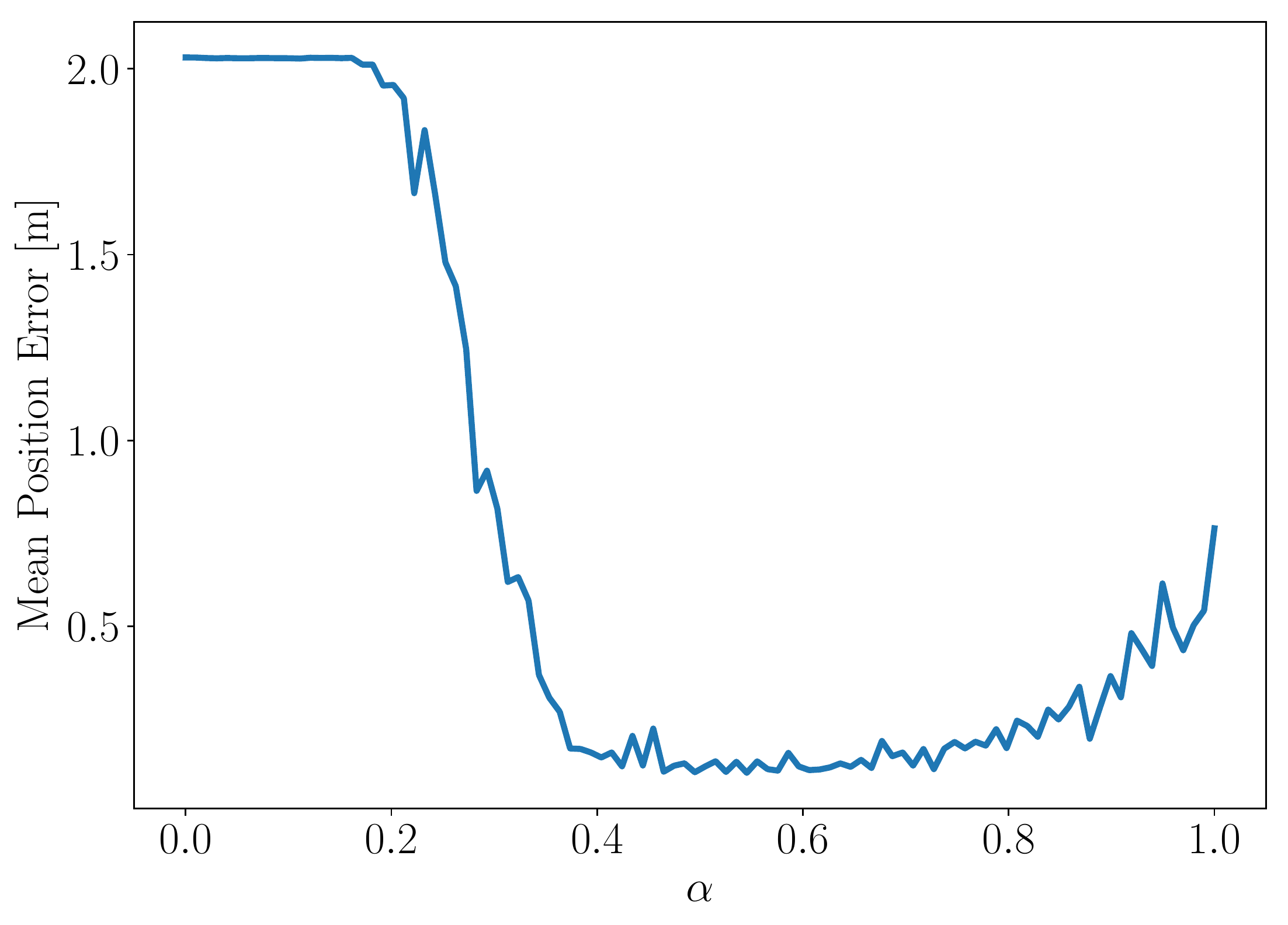}
        \caption{\ninept Mean positions error for different levels of dithering.}
        \label{fig:alpha}
    \end{subfigure}
    \hspace{2mm}
    \begin{subfigure}[b]{0.4\textwidth}
        \includegraphics[width=\textwidth]{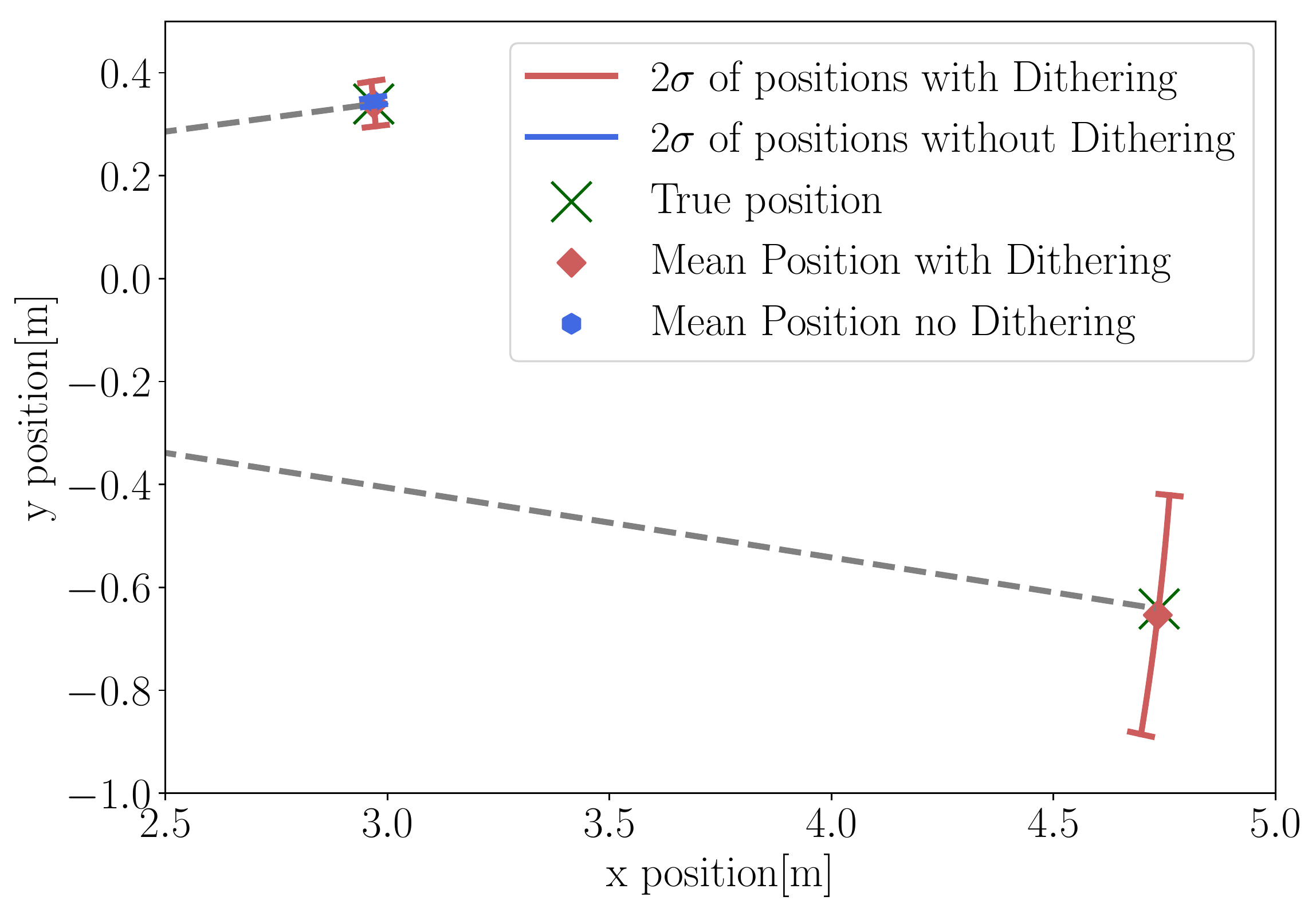}
        \caption{\ninept Reconstructions achieved with weighted dithering.}
        \label{fig:meas_loc}
    \end{subfigure}\vspace{-4mm}
    \caption{\ninept Reconstruction of real measurements.\sq}\label{fig:animals}
    \label{meas}
\end{figure}
Fig.~\ref{fig:meas_loc} shows the reconstruction achieved using the optimal scaling $\alpha=0.55$ of the dithering versus the absence of dithering ($\alpha=0$). The radar is located in $(0,0)$. The variations in the estimations for 106 consecutive measurements are represented by the two sigma span around the mean estimated positions. As already hinted in Fig.~\ref{fig:alpha}, the non-dithered scheme exhibits poor performances as it is not able to resolve the second target at the $4.8$m range. The 1-bit non-dithered quantization has effectively removed the second target from the signal. Adding the weighted dithering allows the recovery of the two targets consistently but at a price in the variance of the closest one. This result shows nonetheless a promising gain of the use of dithering on real applications where noise is encountered.

\hsq
\section{Conclusion}

In this work, we have studied the 2D-localization of multiple targets configurations by using two receiving antennas combined with 1-bit radar quantization, which resulted into the QCS model \eqref{eq:QRadar-system}. We proved that the PBP algorithm for the 2D-localization of targets achieves a bounded reconstruction error decaying as the number of measurements increases. This decaying reconstruction error of the PBP algorithm was further verified using the Monte Carlo simulations and real radar measurements. In particular, the real radar measurements experiments with the radar sensor shed light on the interaction between the system noise and the uniform dithering. Furthermore, we showed how some deterministic artifacts vanishes when a random dithering vector is added in the quantization process. In future works, a deeper study of the range profile estimation using quantized dithered 1-bit radar measurements will provide further insights on how to improve the performances of real radar applications.

\appendix

\hsq
\section{Proof of Prop.~\ref{prop:lpd-existence}}

Extending the LPD to real mappings, we first note that $\bar\Amap \in {\rm LPD}(\Sigma^{2N}_{2K}, \bar{\bs \Phi}, \nu)$ involves $\Amap \in {\rm LPD}(\bar{\Sigma}^N_K, \bs \Phi, 4\nu)$,
with ${\bar\Amap(\bs u) := \cl Q(\bar{\bs \Phi} \bs u + \bs \xi)}$, $\bar{\bs \Phi} := (\bs \Phi_{\rm R}, \bs \Phi_{\rm I}) \in \bb R^{M\times 2N}$ and $\Sigma^{2N}_{2K} := \bar{\Sigma}^{2N}_{2K} \cap \bb R^{2N}$. Indeed, for all $\bs u \in \bar{\Sigma}^N_K$, defining $\bs u^{\rm R} := (\bs u^\top_{\rm R}, - \bs u^\top_{\rm I})^\top \in \Sigma^{2N}_{2K}$ and $\bs u^{\rm I} := (\bs u^\top_{\rm I}, \bs u^\top_{\rm R})^\top \in \Sigma^{2N}_{2K}$, we have $\bs \Phi \bs u = \bar{\bs \Phi} \bs u^{\rm R} + \im \bar{\bs \Phi} \bs u^{\rm I}$ and $\Amap(\bs u) = \bar\Amap(\bs u^{\rm R}) + \im \bar\Amap(\bs u^{\rm I})$.  Therefore, if $\bar\Amap \in {\rm LPD}(\Sigma^{2N}_{2K}, \bar{\bs \Phi}, \nu)$, $|\scp{\Amap(\bs w)}{\bs \Phi \bs v} - \scp{\bs \Phi \bs w}{\bs \Phi \bs v}| \leq \sum_{r,s \in \{``{\rm R}", ``{\rm I}"\}} |\scp{\bar\Amap(\bs w^r)}{\bar{\bs \Phi} \bs v^s} - \scp{\bar{\bs \Phi} \bs w^r}{\bar{\bs \Phi} \bs v^s}| \leq 4 \nu M$.

Interestingly, provided $\sinv{\sqrt M}\|\bar{\bs \Phi}(\bs u - \bs v)\| \leq L \eta$ as soon as $\|\bs u - \bs v\|\leq \eta$ for any $\eta >0$, $L=O(1)$ and $\bs u,\bs v \in \Sigma^{2N}_{2K}$ (\ie if $\frac{1}{\sqrt M}\bar{\bs \Phi}$ is $(\eta,L)$-Lipschitz over $\Sigma^{2N}_{2K}$), \cite[Prop. 6.5]{XJ2018} proves that \whp ${\bar \Amap \in {\rm LPD}\big(\Sigma^{2N}_{2K}, \bar{\bs \Phi}, \nu = \epsilon (1+\delta)\big)}$ provided $M \geq C \epsilon^{-2} K \ln(\frac{N}{K})\ln(1 + \frac{c}{\epsilon^3})$ for some constants $C,c>0$. However, for all $\bs u := (\bs u_1^\top, \bs u_2^\top)^\top ,\bs v := (\bs v_1^\top, \bs v_2^\top)^\top \in \Sigma^{2N}_{2K}$, if $\sinv{\sqrt M}\bs \Phi \in {\rm RIP}(\bar{\Sigma}^N_{2K},\epsilon)$, then $\tinv{2}\|\bar{\bs \Phi}(\bs u - \bs v)\|^2 \leq \|\bs \Phi_{\rm R} (\bs u_1 - \bs v_1)\|^2 + \|\bs \Phi_{\rm I} (\bs u_2 - \bs v_2)\|^2 = \|\bs \Phi (\bs u_1 - \bs v_1)\|^2 + \|\bs \Phi (\bs u_2 - \bs v_2)\|^2 \leq M\, (1+\epsilon) \|\bs u - \bs v\|^2$ since $\bs u_i, \bs v_i \in \Sigma^N_{2K} \subset \bar{\Sigma}^N_{2K}$, which shows that $\sinv{\sqrt M}\bar{\bs \Phi}$ is $(\eta, 4)$-Lipschitz over $\Sigma^{2N}_{2K}$ for any $\eta >0$.  This concludes the proof.


\footnotesize

\end{document}